\documentclass[11pt]{article}
\usepackage[style=authoryear,maxbibnames=10,sorting=none,backend=biber]{biblatex}
\PassOptionsToPackage{margin=1in}{geometry}
\usepackage[english]{babel}
\usepackage[dvipsnames,svgnames,x11names]{xcolor}
\usepackage{adjustbox}
\usepackage[ruled]{algorithm2e}
\usepackage[letterpaper,top=1in,bottom=1in,left=1in,right=1in,marginparwidth=50pt,headsep=12pt,footskip=0.5in,heightrounded]{geometry}

\usepackage{epigraph}
\usepackage{glossaries}
\usepackage[T1]{fontenc}
\usepackage{amsmath}
\usepackage{amssymb}
\usepackage{amsthm}
\usepackage{booktabs}
\usepackage{csquotes}
\usepackage{enumitem}
\usepackage{fancyhdr}
\usepackage{framed}
\usepackage{graphicx}
\usepackage{bm}
\usepackage{makecell}
\usepackage{mathtools}
\usepackage{multicol}
\usepackage{fix-cm}
\usepackage{newpxtext}
\usepackage{newpxmath}
\usepackage{thm-restate}
\usepackage{thmtools}
\usepackage{tikz}
\usepackage{titlesec}
\usepackage{glossary-mcols}
\usepackage{dsfont}
\usepackage{appendix}
\usepackage{siunitx}
\usepackage{caption}
\usepackage{subcaption}
\usepackage{pgf-umlsd}
\usepackage{silence}
\usepackage{hyperref}
\usepackage{cleveref}

\theoremstyle{plain}
\newtheorem{theorem}{Theorem}
\crefname{theorem}{theorem}{theorems}
\Crefname{theorem}{Theorem}{Theorems}

\crefname{lemma}{lemma}{lemmas}
\Crefname{lemma}{Lemma}{Lemmas}

\crefname{claim}{claim}{claims}
\Crefname{claim}{Claim}{Claims}

\crefname{corollary}{corollary}{corollaries}
\Crefname{corollary}{Corollary}{Corollaries}

\newtheorem{proposition}{Proposition}
\crefname{proposition}{proposition}{propositions}
\Crefname{proposition}{Proposition}{Propositions}

\theoremstyle{remark}

\crefname{fact}{fact}{facts}
\Crefname{fact}{Fact}{Facts}

\crefname{remark}{remark}{remarks}
\Crefname{remark}{Remark}{Remarks}

\crefname{example}{example}{examples}
\Crefname{example}{Example}{Examples}

\crefname{conjecture}{conjecture}{conjectures}
\Crefname{conjecture}{Conjecture}{Conjectures}

\theoremstyle{definition}

\crefname{assumption}{assumption}{assumptions}
\Crefname{assumption}{Assumption}{Assumptions}

\newtheorem{definition}{Definition}
\crefname{definition}{definition}{definitions}
\Crefname{definition}{Definition}{Definitions}

\crefname{problem}{problem}{problems}
\Crefname{problem}{Problem}{Problems}

\setkeys{Gin}{width=\linewidth}

\usetikzlibrary{external,arrows,shapes,graphs,positioning}

\newcommand{\1}{\mathds{1}}

\newcommand{\E}{\mathbb{E}}

\newcommand{\R}{\mathbb{R}}

\newcommand{\conv}{\operatorname{conv}}

\newcommand{\majority}{\textsc{maj}}
\newcommand{\minority}{\textsf{min}}

\renewcommand{\P}{\mathbb{P}}

\title{Preference Measurement Error, Concentration in Recommendation Systems, and Persuasion
\thanks{I thank Chara Podimata and Dylan Hadfield-Menell for many discussions about recommendation systems, Sendhil Mullainathan for nudging me to think about the role of noisy observations in Economics more broadly, and seminar audiences at MIT and CMU for helpful feedback.}
}
\author{Andreas Haupt\thanks{Stanford University}}
\date{\today}
\begin{document}
\maketitle
\begin{abstract}
Algorithmic recommendation based on noisy preference measurement is prevalent in recommendation systems. This paper discusses the consequences of such recommendation on market concentration and inequality. Binary types denoting a statistical majority and minority are noisily revealed through a statistical experiment. The achievable utilities and recommendation shares for the two groups can be analyzed as a Bayesian Persuasion problem. While under arbitrary noise structures, effects on concentration compared to a full-information market are ambiguous, under symmetric noise, concentration increases and consumer welfare becomes more unequal. We define \emph{symmetric statistical experiments} and analyze persuasion under a restriction to such experiments, which may be of independent interest.
\end{abstract}

\section{Introduction}\label{sec:introduction}
Personalized experiences are ubiquitous in our everyday lives. From movie recommendations (e.g., Netflix and Hulu) to micro-blogs (e.g., TikTok, X, and Mastodon) and e-commerce (e.g., Amazon and Mercado Libre), people turn to these recommendation systems to select entertainment, information, and products. For example, a recent study by \textcite{gomez2015netflix} revealed that $80\%$ of the approximately $160$ million hours of video streamed on Netflix were recommended by the service's recommendation system.

Personalized experiences are optimized based on feedback from consumers. Algorithms choose what consumers see, and in response, the consumer decides whether and how to engage with this recommended content. From this engagement, the system receives information about the consumer's preferences in the form of likes, comments, and shares, as well as information on what the consumer engages with. In a classical recommendation system, a consumer's recommendations are based on signals $\sigma$, 
\[
\sigma (\theta) = \theta + \varepsilon(\theta)
\]
where $\theta \in \mathbb R^d$ is a consumer type, and $\varepsilon(\theta) \sim F \in \Delta(\R^d)$ is a measurement error. A personalization algorithm will select content $x \in X$ to maximize consumer welfare $u(x; \theta)$, while undoing the effect of noise.

We are interested in the effect of such measurement noise on concentration equality. Consider two types $\theta_\majority$ and $\theta_\minority$, which each have a preferred outcome $x_\majority$ resp. $x_\minority$. There are less minority agents, $\P [\theta_\minority] = \alpha < \frac12$. This is the simple setting of the paper, and allows us to reason about the impacts of preference measurement noise. How does $\alpha$ compare to $\P[x = x_\minority]$? How does utility compare under noise and no noise? We will close the model below, and fully characterize each of them.

One might think that recommendation will favour statistical majorities, as under noise the prior will be more highly weighted, favouring the majority type, hence increasing concentration and inequity. The main observation of this paper is to add nuance to this: The intuition is correct if $\varepsilon (\theta)$ does not depend on $\theta$ (that is, errors are homoskedastic), concentration and inequality in consumer welfare are increased through preference measurement noise. However, under general noise structures, the intuition is misleading. Specifically, even under binary utilities and a general noise structure content may be recommended for any fraction from zero until twice the minority prevalence. Minority welfare might be higher than majority welfare. The main assumption is that minority preference measurements may be significantly more accurate than preference measurements for the majority.

\subsection*{Related Work}\label{subsec:relatedlit}
Our literature relates to the popularity bias in recommendation systems. Popularity bias  \parencite{abdollahpouri2019popularity} is a statistical bias arising from not correcting for propensity in recommending content. The problem of recommending based on little data on the consumer is called the \emph{cold start} and is solved with active exploration techniques
\parencite{safoury2013exploiting, zheng2017identification}. Some papers explicitly consider recommendations for small statistical minorities, which are called \emph{grey sheep consumers} \parencite{alabdulrahman2021catering,zheng2017identification}. Additionally, the present work is related to studies of personalization where consumers choose different strategies to improve their recommendation to a platform \parencite{eslami2016first, lee2022algorithmic, klug2021trick,simpson2022tame,strategic,cen2023user,cen2024measuring}. We take measurement noise as a given, and consider the impacts on market concentration and consumer utilities.

We also relate to a literature in industrial organization on recommendation systems. The simulation study \textcite{calvano2023artificial} considers a two-sided market with a personalization algorithm. The paper's simulations feature significant measurement error---\textcite[Equation (5)]{calvano2023artificial} defines noise due to measurement error twice as large as their heterogeneity among consumers. More broadly, our study can be seen as contributing to a conversation on mass \emph{vs.} niche content. \textcite{Anderson2006longtail} argues that algorithms help a long tail, that is, very infrequently bought, items to rise to prominence. \textcite{fleder2009blockbuster} takes the opposite perspective and points out additional concentration. We conclude that measurement error may or may not increase concentration, depending on how symmetric the noise is.

We also relate to a literature in game theory and behavioral economics. Similarly, the notion of quantal response equilibrium (QRE), \textcite{mckelvey1995quantal} relies on players observing a utility shock, optimizing based on it, but not observing other agent's utility shocks. Our model does not consider random utility shocks, but preference measurement error. Our model also relates to models of mechanism design with complex statistical types \parencite{cai2021recommender,parkes1999accounting,parkes2000iterative}. In contrast to these models, we do not allow the algorithm to decide on queries to the consumer, but we take measurement error as a primitive of the environment. Finally, our preference measurement error can be interpreted as a behavioral imperfection, and actions from consumers to improve signaling as sophisticated behavior, compare \textcite{laibson1997,o1999doing,o2001choice}.

Finally, this work relates to algorithmic fairness. In particular, our comparison of the market share of minority content compared to the minority's incidence is mathematically equivalent to the recommendation algorithm's calibration gap \parencite{Pleiss2017a}. Our result on the incompatibility of fairness and efficiency, \Cref{thm:inversion-concentration}, can hence be interpreted as an instance of the incompatibility of accuracy and calibration. To achieve calibration in recommendation, \textcite{steck2018calibrated} formalizes \emph{item-level} calibration. Recommendations are item-level calibrated if the consumer sees items in a proportion that they consumed them in the past. Our results differ in that we consider the \emph{population-level} distribution of recommendation.

Our techniques make use of (and our Propositions have corresponding results in) the literature on information design and Bayesian persuasion, compare \textcite{Baskerville2011}. In fact, our model is mathematically equivalent to the classical Bayesian Persuasion model \parencite{kamenica2011bayesian}, but with a very different interaction.\footnote{The correspondence is to identify the majority type with the innocent state of the world, the minority type with the guilty state of the world, allocation of majority content with the acquittal action, allocation of minority content with a conviction action, measurement error with the investigation, the judge with the personalization algorithm, and the imagined adversary with the prosecutor.} Persuasion with symmetric statistical experiments has, to the best of our knowledge, not been studied in the information design literature.

\section{Model}\label{sec:inversion-model}
There are two types of agents $\theta_{\minority}, \theta_{\majority}$. A minority $\alpha < \frac12$ is of type $\theta_{\minority}$, a majority $1-\alpha$ is of type $\theta_{\majority}$. We denote the set of types by $\Theta$ and probability distribution on types by $F \in \Delta(\Theta)$. A benevolent personalization algorithm chooses $x \in X \coloneqq \{x_{\minority}, x_{\majority}\}$ for the consumers based on a noisy observation of the type, distributed as $\sigma \colon \Theta \to \Delta(S)$. We will use $S = \R$ for our comparative statics in noise levels. A concrete example we consider is a Gaussian measurement
\[
\sigma(\theta_j) \sim N(\mu_{j}, \kappa^2)
\] 
for $j = \minority,\majority$. The algorithm wishes to maximize consumer welfare, which is binary,
\[
u(x_j, \theta_{j'}) = \begin{cases}
    1 & j = j'\\
    0 & \text{else.}
\end{cases}
\]
The timeline of the interaction is as follows. First, the consumer's type $\theta\sim F$ is realized. Next, the observation $s \sim \sigma(\theta)$ is realized and observed by the algorithm. Finally, the algorithm chooses $x \in X$ to maximize consumer welfare.

\section{The Consequences of Measurement Error}\label{sec:concentration}
In this section, we investigate the probability of allocating minority content, $\P[x_\minority]$, which we will call minority share, and the utilities of majority and minority consumers. A low probability of recommending minority content means a high amount of concentration.

\subsection{Market Concentration}
We first investigate how likely it is that minority content is allocated, $\P[x_\minority]$, which can be interpreted as a \emph{minority content market share}, or short, \emph{minority share}. A clear standard is the incidence of the minority in the population. For example, if a statistical minority is 10\% of the consumer population, we are interested in whether an optimal personalization algorithm will recommend content more, or less, than with 10\% probability. We first show that under general measurement error, it is possible that the minority anywhere from not at all to twice the minority incidence (i.e., in our example, $20\%$) is possible. It must be (weakly) less than the minority incidence for \emph{symmetric} measurement error.

\begin{figure}
\centering
\begin{tikzpicture}[scale=5]
\draw[->] (-0.1, 0) -- (1.1, 0) node[right] {$\mu$};
\draw[->] (0, -0.1) -- (0, 1.1) node[above] {$x_\minority$ recommended};
\draw (0.5, -.05) -- (0.5, .05);
\draw (1, -.05) -- (1, .05);
\draw (-.05, .5) -- (.05, .5);
\draw (-.05, 1) -- (.05, 1);

\draw (0.5, -.05) node[below] {$0.5$};
\draw (1, -.05) node[below] {$1$};
\draw (-.05, 2*0.25) node[left] {$2\alpha$};
\draw (-.05, 1) node[left] {$1$};

\draw[thick] (0, 0) -- (0.5, 0);
\draw[thick] (0.5, 1) -- (1, 1);

\draw[thick, dashed] (0, 0) -- (0.5, 1);

\draw[dashed] (0.25, 0) -- (0.25, 0.5);
\draw[dashed] (0, 0.5) -- (0.25, 0.5);
\fill (0.25, 0.5) circle (0.5pt) node[right] {$(\alpha, 2\alpha)$};
\end{tikzpicture}
\caption{The information structure maximizing minority content allocation.}
\label{fig:inverrsion-concave-closure}
\end{figure}
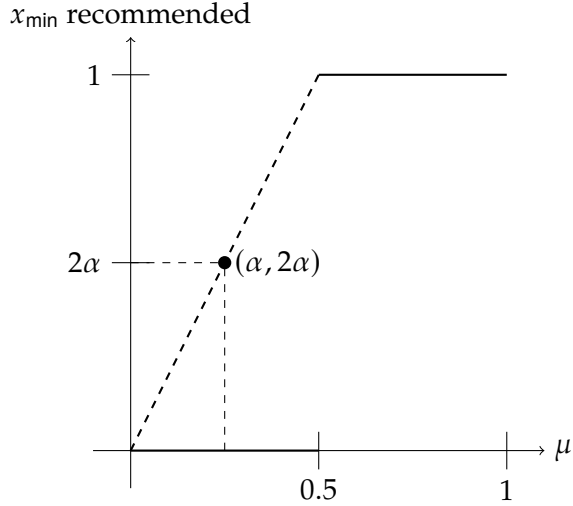

We first start with a property for general measurement error: It is possible that the minority share is anywhere from zero to twice the minority incidence. Any share in between is achievable by a measurement error.

\begin{proposition}\label{prop:inversion-extreme-info}
For any $p \in [0, 2\alpha]$, there is a measurement structure $\sigma$ such that $\P[x_\minority] = p$. For any measurement structure, $\P[x_\minority] \le 2 \alpha$.
\end{proposition}

The measurement errors that achieve the extreme cases of minority share are intuitive. First, consider the case in which the measurement error is so large that it is pure noise. In this case, optimal personalization will serve the majority content as it is more liked in the population. Hence, the minority share is zero under this measurement error.

The measurement errors that lead to higher-than-incidence minority share feature some asymmetry in their informativeness. While the minority incidence is, definitionally, lower than the majority incidence in the population, it is possible that the minority incidence \emph{conditional on a signal} is higher than the majority incidence. In this case, consumers with such a signal will be allocated minority content. The extreme case achieving $2\alpha$ minority share is the case where conditional on all noisy observations of the minority type, the minority is in the majority conditional on the signal.

\begin{proof}
    We use techniques from Bayesian persuasion \parencite{Kamenica2019a}. We can view this problem as a setting where a sender chooses a preference measurement error structure $\sigma$, that is, a statistical experiment. The statistical experiment $\sigma$ induces posterior probabilities over $\theta_{\minority}$ and $\theta_{\operatorname{max}}$, which we can identify with a posterior distribution $\P[\theta_{\minority} | s]$ for $s \sim \sigma(\theta)$. Denote the distribution of these posteriors by $\mu_\sigma$. We have that
    \[
        \P[x_\minority] = \E_{x \sim \mu_\sigma} [ \1_{x \ge \frac12}]
    \]
    We can use \textcite[Proposition 1]{kamenica2011bayesian} to reduce this problem to choosing a posterior that, on average, is the prior, i.e. $\E[\mu_\sigma] = \alpha$. Such posteriors are also called \emph{Bayes-plausible} posteriors in the literature following \textcite{kamenica2011bayesian}. We can hence construct Bayes-plausible posteriors for the first part of the statement. To this end, we consider the Bayes-plausible posteriors that put mass $p$ on $\frac12$, and mass $1-p$ on $\frac{\alpha - \frac{p}{2}}{1-p}$. This is a Bayes-plausible distribution that achieves mass $\E_{x \sim \mu_\sigma} [ \1_{x \ge \frac12}] = p$. That is, probability $p$ for classifying as $\theta_\minority$. 

    For the second part of the statement, we use \textcite[Corollary 1]{kamenica2011bayesian}: The concave closure of the function $\1_{x \ge \frac12}$, which is
    \begin{equation}
    \widehat{\1_{x \ge \frac12}} = 
        \begin{cases}
            2x & x \in [0, \frac12]\\
            1 & x \in (\frac12, 1]
        \end{cases}
    \end{equation}
    evaluated at $\alpha$, hence $2\alpha$, see \Cref{fig:inverrsion-concave-closure}.
\end{proof}

In many environments, however, one does not expect such significant asymmetry, and we may have some structure where the observation probability of some signal conditional on types are the same. For this, we call a function $l \colon S \to S$ an \emph{involution} if $l(l(s)) = s$ for all $s \in S$. The main example of an involution  we consider is the \emph{midplane reflection}. It is defined by 
\[
l(s) =  \left(s - 2 \frac{(s - \mu_\minority) \cdot (\mu_\majority - \mu_\minority)}{\|\mu_\majority - \mu_\minority\|^2} (\mu_\majority - \mu_\minority) \right),
\]
and depicted in \Cref{fig:inversion-midplane-reflection}. The existence of an intraversion means that there naturally are pairs $(s, l(s))$ in the signal space. (Note that $s = l(s)$ is possible, and that the identity is an involution.)

Having an intraversion $l$, we can define symmetry of measurement error.
\begin{definition}\label{def:symmetric}
    We say that a measurement system $\sigma \colon \Theta \to \Delta(S)$ is \emph{symmetric} if the density from the minority for a point $s$ is the same as for the majority at $l(s)$,
    \begin{equation}
\sigma(\theta_\minority) (s) = \sigma (\theta_\majority) (l(s)).\label{eq:inversion-symmetry}
    \end{equation}
\end{definition}
By definition of an involution, a symmetric measurement error also satisfies $\sigma(\theta_\minority) (l(s)) = \sigma (\theta_\majority) (s)$. A main example of symmetric measurement error for the midplane reflection $l$ is $\sigma(\theta_j) = N(\mu_j, \kappa^2)$ for $\mu_j \in \R$ and some \emph{common} variance $\kappa^2$ of two distributions.

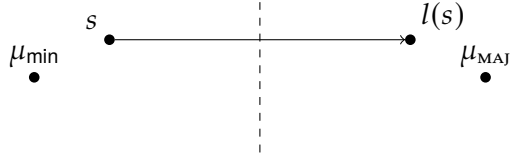
\begin{figure}
\centering
\begin{tikzpicture}

\coordinate (mu_min) at (0,0);
\coordinate (mu_maj) at (6,0);
\coordinate (equidistant_line_start) at (3,-1);
\coordinate (equidistant_line_end) at (3,1);
\coordinate (s) at (1, 0.5); 
\coordinate (s_reflection) at (5, 0.5); 

\draw[dashed] (equidistant_line_start) -- (equidistant_line_end);

\fill (mu_min) circle (2pt) node[above] {$\mu_{\minority}$};
\fill (mu_maj) circle (2pt) node[above] {$\mu_{\majority}$};

\fill (s) circle (2pt) node[above left] {$s$};
\fill (s_reflection) circle (2pt) node[above right] {$l(s)$};
\draw[->, shorten <=2pt, shorten >=2pt] (s) -- (s_reflection);
\end{tikzpicture}
\caption{The midplane reflection.}
\label{fig:inversion-midplane-reflection}
\end{figure}

\begin{theorem}\label{thm:inversion-concentration}
    Assume that a measurement system is symmetric. Then $\P[x_\minority] \le \alpha$.
\end{theorem}
Hence, under symmetric measurement error, the minority share is lower than minority incidence.

Note that \Cref{thm:inversion-concentration} holds for any involution $l$, in particular the identity function. For the identity function, $\P[\theta_\majority | s] = 1-\alpha > \alpha = \P[\theta_\minority | s]$ for any $s \in S$, and hence $\P[x_\majority] = 1$ and $\P[x_\minority] = 0$. Hence, the minority share is zero in this case.

The intuition of the proof is to show that the signal pairs $(s, l(s))$ implied by the involution cannot both lead to the allocation of minority content. While this alone is not enough to conclude, we show that if one of them is minority type, then the majority type must have strictly higher likelihood.

\begin{proof}
    Denote the set of signals $s \in S$ that are served $x_\minority$ by $S_\minority$. That is,
    \begin{equation}
    S_\minority = \left\{ s \in S : \P [ \theta_\minority | s] \ge \frac12\right\}.\label{eq:inversion-isotropy-classification}
    \end{equation}
    Define $S_\majority = S \setminus S_\minority$. We will write $\P[s]$ for the likelihood of $s$.

    We first consider points $(s, l(s))$ and observe that at least one of the points must be in $S_\majority$. Then, we show that if $s\in S_\minority$, then $l(s)$ has a higher likelihood than $s$. In a third step, we conclude.

    Let $s\in S_\minority$. By Bayes' rule, it must be the case that
    \[
    \frac{\alpha}{1-\alpha} \frac{\P [ s | \theta_\minority ]}{\P [ s | \theta_\majority]} = \frac{\alpha}{1-\alpha} \frac{\sigma( \theta_\minority)(s)}{\sigma( \theta_\majority)(s)} \ge 1.
    \]
    Hence, as $\frac{\alpha}{1-\alpha} < 1$, it must be that $\frac{\P [ s | \theta_\minority ]}{\P [ s | \theta_\majority]} > 1$. Hence, by symmetry, $\frac{\P [ l(s) | \theta_\majority ]}{\P [ l(s) | \theta_\minority]} > 1$. As $\frac{1-\alpha}{\alpha} > 1$, 
    \[
     \frac{1-\alpha}{\alpha} \frac{\P [l(s) | \theta_\majority ]}{\P [ l(s) | \theta_\minority]} > 1,
    \]
    and $l(s) \in S_\majority$. Next we show, that if $s \in S_\minority$ and $l(s) \in S_\majority$, then $\P[s] \le \frac{\alpha}{1-\alpha} \P[l(s)]$. This follows from the following chain of inequalities:
\begin{align*}
    \P[s] &= \alpha \sigma(\theta_\minority) (s)+ (1-\alpha) \sigma(\theta_\majority) (s) \\
    &\le \alpha \sigma(\theta_{\minority}) (s) + \alpha \sigma(\theta_{\minority}) (s) \\
    & = \alpha \sigma(\theta_{\majority}) (l(s)) + \alpha \sigma(\theta_{\majority}) (l(s)) \\
    & \le \frac{\alpha^2}{1-\alpha} \sigma(\theta_{\minority}) (l(s)) + \alpha \sigma(\theta_{\majority}) (l(s)) \\
    & = \frac{\alpha}{1-\alpha} \alpha\sigma(\theta_{\minority}) (l(s)) + \frac{\alpha}{1-\alpha} (1-\alpha) \sigma(\theta_{\majority}) (l(s)) \\
    & = \frac{\alpha}{1-\alpha} [ \alpha \sigma(\theta_{\minority}) (l(s)) + (1-\alpha) \sigma(\theta_{\majority}) (l(s))]\\
    &= \frac{\alpha}{1-\alpha} \P[l(s)].
\end{align*}
    The two inequalities use that $s \in S_\minority$ and $l(s) \in S_\majority$.
    
    Hence, we can decompose $S = S_\text{both} \cup S_\minority \cup l(S_\minority)$, where $S_\text{both} = \{ s \in S | s, l(s) \in S_\majority\}$. Note that this is a disjoint union. We hence have that 
    \[
    \P[S_\minority] \le  \frac{\alpha}{1-\alpha} \P[l(S_\minority)] \le  \frac{\alpha}{1-\alpha} ( 1- \P[S_\minority]).
    \]
    Algebra shows that this inequality implies $\P[S_\minority] \le \alpha$. As $\P[ x_\minority ] = \P[S_\minority]$, this concludes the proof.
\end{proof}
Hence, symmetric measurement error increases concentration. For a particular model of noise, we can show that not only is the minority share lower under measurement error than without measurement error, but also that it is \emph{decreasing} in measurement error.

Consider one-dimensional Gaussian measurement error
\[
\sigma_\kappa(\theta_j) = N(\mu_{j}, \kappa^2).
\]
It is without loss to normalize $\mu_{\minority} = 0$ and $\mu_{\majority} = 1$. For such Gaussian measurement error, the minority share decreases as measurement error increases.
\begin{proposition}\label{prop:market-concentration-gaussian}
For not too large $\kappa < (\ln (\frac{\alpha}{1-\alpha}))^{-\frac12}$, $\P_{\sigma_\kappa}[ x_\minority]$ is monotonically non-increasing in $\kappa$.
\end{proposition}

The condition on $\kappa$ is mild. It holds for decision boundaries $x^* \le -\frac32$, far on the left of the minority type, meaning a rather extreme level of preference measurement error.

\begin{proof}
The decision boundary is given by an equality of likelihood:
\[
\alpha \cdot \frac{1}{\kappa \sqrt{2\pi}} e^{-\frac{x^{2}}{2\kappa^2}} = (1 - \alpha) \cdot \frac{1}{\kappa \sqrt{2\pi}} e^{-\frac{(x - 1)^2}{2\kappa^2}}.
\]
Algebra yields that the decision boundary is
\[
x^* = \frac{1}{2} + \kappa^2 \ln\left(\frac{\alpha}{1 - \alpha}\right).
\]
All signals $s \le x^*$ are allocated $x_{\minority}$, all other signals are allocated $s_{\majority}$. The probability that minority content is allocated is
\[
\alpha \Phi\left(\frac{1}{2\kappa} + \frac{\kappa \ln\left(\frac{\alpha}{1 - \alpha}\right)}{2}\right) + (1 - \alpha) \Phi\left(-\frac{1}{2\kappa} + \frac{\kappa \ln\left(\frac{\alpha}{1 - \alpha}\right)}{2}\right).
\]
Here, $\Phi$ is the cumulative distribution function of a standard Gaussian. The first summand is the contribution of allocating for minority types, the second summand is for majority types. It is again a result of algebra that the derivative of this function with respect to $\kappa$ is
\[
 \alpha \phi\left(\frac{1}{2\kappa} + \frac{\kappa \ln\left(\frac{\alpha}{1 - \alpha}\right)}{2}\right) \left(-\frac{1}{2\kappa^2} + \frac{\ln\left(\frac{\alpha}{1 - \alpha}\right)}{2}\right) + (1 - \alpha) \phi\left(-\frac{1}{2\kappa} + \frac{\kappa \ln\left(\frac{\alpha}{1 - \alpha}\right)}{2}\right) \left(\frac{1}{2\kappa^2} + \frac{\ln\left(\frac{\alpha}{1 - \alpha}\right)}{2}\right),
\]
where $\phi$ is the probability density function of a standard Gaussian. As probability density functions are non-negative, and by assumption on $\kappa$, this function is negative, as required.
\end{proof}
This section contained three results: First, under general measurement errors, the minority share may be zero, smaller than the incidence of the minority, or up to twice as large. When measurement error is symmetric, minority content must be (weakly) less allocated than minority incidence in the population. Hence, minority content is disadvantaged compared to its incidence in the population. Finally, for (not too large) Gaussian measurement error, the minority share is decreasing in measurement error.

\subsection{Consumer Welfare}
As in the previous subsection, we first show the consumer utilities possible under under arbitrary measurement error structures.
\begin{proposition}
    The set of achievable consumer utilities for majority and minority under arbitrary measurement error is the triangle $\conv (\{(0, 1), (1, 1), (1, \frac{1-2\alpha}{1-\alpha})\}$.
\end{proposition}

Note that the triangle of achievable utilities for the minority and the majority is tilted. While minority utilities of zero (hence no minority agent is served their preferred content) are possible, there is a lower bound for the majority.

\begin{proof}
    Note that by definition of the utilities, it must be that the set of achievable utilities is contained in $[0,1]^2$. In addition, the set of achievable utilities is convex, compare \textcite{kamenica2011bayesian}. As in the proof of \Cref{thm:inversion-concentration}, it is sufficient to (a) show that there are Bayes-plausible posterior distributions that yield utilities $(0, 1)$, $(1, 1)$, and $(1, \frac{1-2\alpha}{1-\alpha})$ and (b) show that $\frac{\alpha}{1-\alpha} u_\minority + u_\majority \ge 1$. (Note that $\frac{\alpha}{1-\alpha} u_\minority + u_\majority = 1$ is the line through $(0, 1)$ and $(1, \frac{1-2\alpha}{1-\alpha})$, compare \Cref{fig:inversion-triangle}.)

    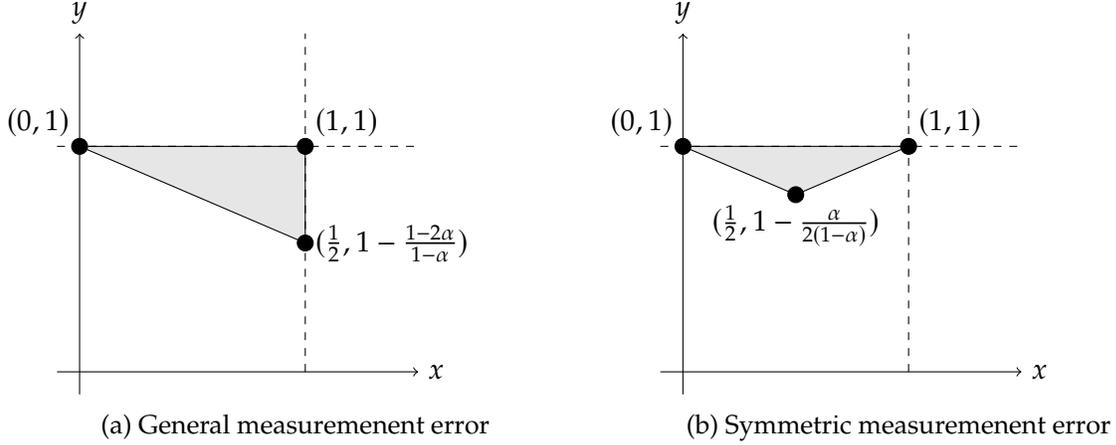
\begin{figure}
        \centering
\begin{subfigure}[t]{.48\linewidth}

\begin{tikzpicture}[scale=3]

    \coordinate (A) at (0, 1);
    \coordinate (B) at (1, 1);
    \coordinate (C) at (1, {.4/.7});

    \draw[fill=gray!20] (A) -- (B) -- (C) -- cycle;

    \filldraw[black] (A) circle (1pt) node[above left] {$(0, 1)$};
    \filldraw[black] (B) circle (1pt) node[above right] {$(1, 1)$};
    \filldraw[black] (C) circle (1pt) node[right] {$(\frac12, 1-\frac{1-2\alpha}{1-\alpha})$};

    \draw[dashed] (1, 0) -- (1, 1.5);  
    \draw[dashed] (-0.1, 1) -- (1.5, 1);  

    \draw[->] (-0.1, 0) -- (1.5, 0) node[right] {$x$};
    \draw[->] (0, -0.1) -- (0, 1.5) node[above] {$y$};
\end{tikzpicture}
\caption{General measuremenent error}
\label{subfig:inversion-utilities-general}
\end{subfigure}
\begin{subfigure}[t]{.48\linewidth}
\begin{tikzpicture}[scale=3]

    \coordinate (A) at (0, 1);
    \coordinate (B) at (1, 1);
    \coordinate (C) at (1/2, {11/14});

    \draw[fill=gray!20] (A) -- (B) -- (C) -- cycle;

    \filldraw[black] (A) circle (1pt) node[above left] {$(0, 1)$};
    \filldraw[black] (B) circle (1pt) node[above right] {$(1, 1)$};
    \filldraw[black] (C) circle (1pt) node[below] {$(\frac12, 1-\frac{\alpha}{2(1-\alpha)})$};

    \draw[dashed] (1, 0) -- (1, 1.5);  
    \draw[dashed] (-0.1, 1) -- (1.5, 1);  

    \draw[->] (-0.1, 0) -- (1.5, 0) node[right] {$x$};
    \draw[->] (0, -0.1) -- (0, 1.5) node[above] {$y$};
\end{tikzpicture}
\caption{Symmetric measuremenent error}
\label{subfig:inversion-utilities-symmetric}
\end{subfigure}
        \caption[Achievable consumer utilities under measurement error]{Achievable consumer utilities under measurement error.}
        \label{fig:inversion-triangle}
    \end{figure}

    To show (a), we observe that utility profile $(1,1)$ is achieved by a Bayes-plausible posterior of 
    \[
    \mu_\sigma = \begin{cases}
        0 & \text{w.p. } \alpha \\
        1 & \text{w.p. } 1-\alpha,
    \end{cases}
    \]
    which corresponds to no measurement error. Similarly, $(0,1)$ is achieved by the deterministic posterior $\mu_\sigma = \alpha$, corresponding to \enquote{perfect} measurement error, that is, no information. Finally, $(1, 1-\alpha)$ is achieved by the Bayes-plausible posterior
    \[
    \mu_\sigma = \begin{cases}
        \frac{1}{2} & \text{w.p. } 2\alpha \\
        0 & 1-2 \alpha.
    \end{cases}
    \]
    When breaking ties at equal odds for both groups in favor of the majority, this leads to utility 1 for minority group consumers and utility $\frac{1-2\alpha}{1-\alpha}$ for majority consumers.

    To show (b), it must be the case that for every agent who does not win from the majority group, there must be at least an equal mass of agents from the minority group that gets correctly allocated. Formally,
    \begin{multline*}
     \alpha u_\minority = \P[\theta_\minority] \P[x_\minority|\theta_\minority]
     = \P[x_\minority;\theta_\minority] \\\ge \P[x_\minority;\theta_\majority] = \P[\theta_\majority] \P[x_\minority|\theta_\majority] =  (1-\alpha) (1-u_\majority).
    \end{multline*}
    Rearranging, we obtain 
    \[
        \frac{\alpha}{1-\alpha} u_\minority + u_\majority \ge 1
    \]
    as desired.
\end{proof}

For symmetric preference errors, we get a more restricted set of achievable utility profiles:
\begin{theorem}
    The set of achievable consumer utilities for majority and minority under symmetric measurement error is the triangle $\conv (\{(0, 1), (1, 1), (\frac12, 1-\frac{\alpha}{2(1-\alpha)})\}$.
\end{theorem}

Two observations of how symmetry leads to inequality are noteworthy: With symmetry it is impossible for the majority to have lower than perfect utility without leading to lower utility for the minority, and minority utility is decreasing faster than majority utility.

A second observation is about the dependence of the utilities on $\alpha$. For $\alpha$ approaching $0$, the triangle flattens, guaranteeing the majority a high utility.

The proof of this result uses a guess-and-check approach. We first use the bound on minority shares from \Cref{thm:inversion-concentration} as an additional constraint, yielding a candidate set of achievable utilities, and then construct a symmetric measurement error that achieves this utility profile.

\begin{proof}
    Observe that the set of achievable utility profiles is convex also for symmetric measurement errors. Indeed, by mixing the signal assignments, averages of utilities are possible.

    Also note that the utility profiles $(0,1)$ and $(1,1)$ can be achieved with symmetric measurement error (uninformative resp. perfectly informative). A symmetric measurement error achieving $(\frac12, 1-\frac{\alpha}{2(1-\alpha)})$ uses three signals $s, l(s)$ and $\tilde s$ (where $l(\tilde s) = \tilde s$) and is defined as
    \begin{align*}
    \sigma(\theta_\minority)(s) &= \sigma(\theta_\majority)(l(s)) =  \frac12\\
    \sigma(\theta_\majority)(s) &= \sigma(\theta_\minority)(l(s)) =  \frac{\alpha}{2(1-\alpha)}\\
    \sigma(\theta_\minority)(\tilde s) & =  
    \sigma(\theta_\minority)(l(s)) = 1- \frac12 - \frac{\alpha}{2(1-\alpha)}.
    \end{align*}
    In this case, for signal $s$, $x_\minority$ is allocated, for $l(s)$ and $\tilde s$, $x_\majority$ is allocated. This leads to utilities $u_\minority = \frac12$ and $u_\majority = 1-\frac{\alpha}{2(1-\alpha)}$.

    Algebra shows that the additional constraint of the triangle (that is, the left top side of the line through $(\frac12, 1-\frac{\alpha}{2(1-\alpha)})$ and $(1,1)$ is given by $\alpha u_\minority + (1-\alpha)(1-u_\majority) \le \alpha$. As $\alpha u_\minority + (1-\alpha)(1-u_\majority) = \P[\theta_\minority] \P[x_\minority | \theta_\minority] + \P[\theta_\majority] \P[x_\minority | \theta_\majority] = \P[x_\minority]$, this follows from \Cref{thm:inversion-concentration}.
\end{proof}

Finally, we show that for a structured (Gaussian) measurement error, we get that inequality is increasing in measurement error. (This result is independent of the size of the error, in contrast to \Cref{prop:market-concentration-gaussian}.)
\begin{proposition}\label{prop:inversion-value-of-data}
    $u(\theta_\minority)$ is monotonically non-increasing in $\kappa$.
\end{proposition}
\begin{proof}
    As before, the decision boundary is given by
    \[
    x^* = \frac{\kappa^2 \ln\left(\frac{\alpha}{1 - \alpha}\right) + 1}{2}.
    \]
    The probability that an agent of type $\theta_\minority$ is served $x_\minority$ is given by
    \[
    \Phi\left(\frac{1}{2\kappa} + \frac{\kappa \ln(\frac{\alpha}{1 - \alpha})}{2}\right).
    \]
    The derivative of this function with respect to $\kappa$ is
    \[
    \phi\left(\frac{1}{2\kappa} + \frac{\kappa \ln\left(\frac{\alpha}{1 - \alpha}\right)}{2}\right) \left(-\frac{1}{2\kappa^2} + \frac{\ln(\frac{\alpha}{1 - \alpha})}{2}\right).
    \]
    As $\phi$ is a non-negative function and $\ln (\frac{\alpha}{1-\alpha})<0$, this is non-positive.
\end{proof}

\printbibliography

\end{document}